\documentclass[12pt]{article}
\usepackage[final]{epsfig}
\usepackage[english]{babel}
\usepackage{algpseudocode,amsmath,makecell,multicol,authblk,verbatim,amsthm,verbatim,color,graphicx,amssymb,amsmath,latexsym,wrapfig,subfig,multicol,url,siunitx,hyperref,doi,enumitem,xargs,babel,float,soul}
\setcellgapes{8pt}
\usepackage[T1]{fontenc}
\usepackage{amssymb, color}
\usepackage[utf8]{inputenc}
\usepackage[pdftex,dvipsnames]{xcolor}  
\usepackage{booktabs}
\usepackage{algorithm2e}
\usepackage{multirow}

\hypersetup{colorlinks=true,linkcolor=black}

\makeatletter
\def\blfootnote{\xdef\@thefnmark{}\@footnotetext}
\makeatother
\newcommand{\R}{\mathbb{R}}

\newcommand{\bfa}{{\bf a}}
\newcommand{\bfb}{{\bf b}}

\newcommand{\bfe}{{\bf e}}

\newcommand{\bfn}{{\bf n}}

\newcommand{\bfp}{{\bf p}}

\newcommand{\bft}{{\bf t}}

\newcommand{\bfu}{{\bf u}}
\newcommand{\bfv}{{\bf v}}
\newcommand{\bfw}{{\bf w}}

\newcommand{\bfy}{{\bf y}}

\newcommand{\bfz}{{\bf z}}
\newcommand{\bfA}{{\bf A}}

\newcommand{\bfI}{{\bf I}}

\newcommand{\bfM}{{\bf M}}

\newcommand{\bfP}{{\bf P}}

\newcommand{\bfR}{{\bf R}}

\newcommand{\bfT}{{\bf T}}

\newcommand{\bfW}{{\bf W}}

\newcommand{\beq}{\begin{equation}}
\newcommand{\eeq}{\end{equation}}
\newcommand{\beqs}{\begin{eqnarray}}
\newcommand{\eeqs}{\end{eqnarray}}

\newcommand{\calG}{{\cal G}}

\newcommand{\calS}{{\cal S}}

\newcommand{\Rmnum}[1]{\text{\uppercase\expandafter{\romannumeral #1}}}
\newtheorem{theorem}{Theorem}

\newtheorem{remark}{Remark}

\newcommand{\bfig}{\begin{figure}[!h]}	
\newcommand{\efig}{\end{figure}}		

\begin{document}

\begin{center}
{\LARGE  \bf Local growth laws determine global shape of molluscan shells}  \\

\normalsize

\vspace{6mm}
Huan Liu and Kaushik Bhattacharya\footnote{Corresponding author: bhatta@caltech.edu}
\vspace{3mm}

Division of Engineering and Applied Sciences, \\
California Institute of Technology, Pasadena, CA 91125, USA\\

\end{center}

\begin{abstract}

Molluscan shells come in various shapes and sizes.  Despite this diversity, each species produces a shell with a characteristic shape that is independent of environmental conditions. We seek to understand this robust complexity.  We are guided by two principles in the spirit of D'Arcy Thompson.  First, the growth is governed by the repeated and continuous application of a fixed growth law, even as the shell evolves in overall shape,  without any complex biological machinery to monitor and control the growth.  Second, the growth law depends solely on local geometry at the shell's growing edge.  The first principle naturally leads to the mathematical statement that the shape of the shell is generated by the action of a Lie group on a protoconch. The second naturally leads to a particular representation of the Lie group.  We use this representation to show that the shapes of nearly all known molluscan shells can be described by essentially three parameters: a scalar (scaling), a vector (orientation), and a curve (edge of the protoconch).  We relate these parameters to the phylogenetic tree.  In addition to the morphogenetic insight, our results potentially point to a new approach to engineering complex structures.
\end{abstract}


\paragraph{Significance:}  We show that molluscan shells may be described as conformal Lie groups whose parameters have a close relationship with phylogeny. Further, this shape and group are equivalent to a curvature-driven growth law at the growing edge. Finally, the shape and growth law are robust { with respect to variations in growth rates arising from environmental influences. These results suggest that, once biologically programmed through the selection of Lie-group parameters, shell growth can proceed according to local growth laws acting at the growth edge, without requiring active shape control by the mantle.}



\section{Introduction}

The ability to design and program biological growth holds transformative potential for both synthetic biology \cite{scott2010interdependence,klumpp2009growth} and the engineering of intelligent systems \cite{kriegman2020scalable}. However, despite rapid advances in bioengineering and intelligent systems, the underlying principles that govern biological growth remain incompletely understood. Growth can be an emergent process driven by complex, multiscale interactions among genetic regulation, physical forces, and environmental cues.  Understanding these is not only crucial for advancing fundamental biology but also an essential step toward the development of programmable living matter and bioinspired intelligent systems capable of autonomous adaptation and self-organization.

Traditional approaches to morphogenesis \cite{johnson2019growth,boettiger2009neural,serra2020dynamic}, the biological process that causes an organism to develop its shape, have focused on detailed, bottom-up explanations involving specific genetic pathways and biochemical reactions. While these models are critical for understanding cellular-level mechanisms, they often lack a broader, general framework that can account for the geometric regularities observed across biological forms. This limitation presents fundamental challenges: can we move beyond cataloging individual biological processes to uncover a concise, powerful language that governs growth across seemingly disparate forms? Can simple growth rules generate a variety of complex shapes? These are the core problems that motivate this research.  

Among the diverse manifestations of biological growth, molluscan shell formation offers a uniquely informative system for studying the emergence of complex form. Shells grow through surface accretion, guided by the spatiotemporal regulation of secretory activity at the mantle edge \cite{biomineral,marin2007molluscan}. Simply by repeating this accretion process, shells exhibit intricate geometry, hierarchical microstructures, and robust growth dynamics \cite{marin2020mollusc,chirat2021physical}. These properties make shell growth an interesting model for uncovering general principles of morphogenesis and for exploring how growth principles can give rise to diverse forms at larger scales. 

Conventionally, shell growth has been thought to proceed under the guidance of an elaborate template composed of an organic matrix secreted by the mollusk \cite{du2024non}. Recent perspectives downplay the role of the organic matrix and instead view shell formation as a thermodynamically driven physico-chemical process.
This line of thought traces back to Thompson’s seminal work {\it On Growth and Form} \cite{thomson1917growth}, where he observed that a wide range of biological forms follow certain mathematical descriptions. This work itself was influenced by earlier studies by Rainey \cite{rainey1857formation} and Harting \cite{harting1872recherches}, who demonstrated that complex mineral structures can form abiotically in simple colloidal systems. 

More recently, researchers have shown that shell microstructures can naturally arise through classical crystal growth mechanisms \cite{schoeppler2019crystal,best2024classical,granasy2021phase}, implying that the formation of shells may inherently integrate growth principles established in materials science.  Given the diversity of shell morphologies, these material-based local growth laws could be mediated or fine-tuned by biological activities. It's shown that the neurosecretory system in mollusks is highly active at the mantle margin, where it affects shell secretion \cite{boettiger2009neural,li2024neuropeptides}. This indicates that shell growth could be neurally adjusted through feedback from the shell margin. Moreover, the shell growth and mantle growth can be uncoupled \cite{lewis1997growth}, allowing shell morphology to be preserved under varying mantle states. This observation implies that shell formation is not simply a passive outcome of mantle growth, but a self-organising process in mollusks driven by robust growth laws relating to the shell's own state.

\section{Growth law, shell shape, and Lie group}
Motivated by the observations described above, we seek emergent laws to describe the growth of general molluscan shells of intricate diversity without appealing to specific biochemical processes. This top-down approach requires first establishing a formal language to describe their growth kinematics.  We begin with the observation that the shape of the growing edge remains invariant during growth, and traces a trajectory that combines rotation and dilatation  \cite{thomson1917growth} as shown in Fig.~\ref{fig:lie_elements}A.  This suggests that the growth process can be formulated in a formal mathematical framework and that the shells grow by accretion to the growth front.  

\begin{figure}[t]
\centering
\includegraphics[width=0.7\textwidth]{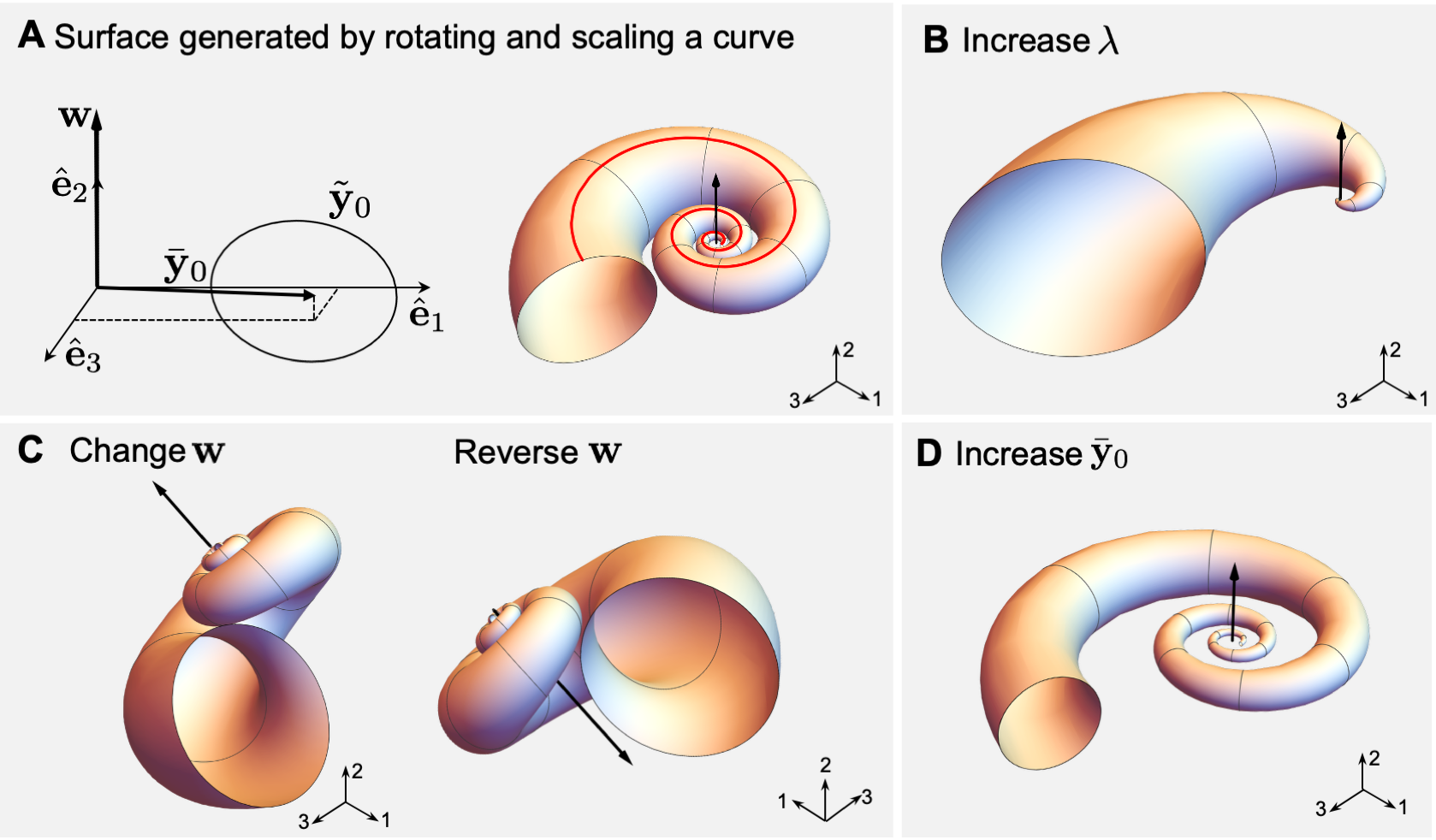}
\caption{ Shells generated by scaling and rotation of the growth front according to (\ref{seashell}). (A) Surface generated by rotating a circular curve $\widetilde{\bfy}_0=\sin x\,\hat{\bfe}_1+\cos x\,\hat{\bfe}_2$, $x\in[0,2\pi]$, with centroid $\overline{\bfy}_0=2.1\,\hat{\bfe}_1-0.5\,\hat{\bfe}_3$ about the axis $\bfw=\hat{\bfe}_2$, while scaling by $\lambda=1.2$ for $t\in[0,17]$. (B) Increasing $\lambda$ to $2$ results in more rapid expansion. (C) Changing the rotation axis to $\bfw=\pm(\hat{\bfe}_1+\hat{\bfe}_2)/\sqrt{2}$ produces a spiral shell; reversing $\bfw$ reverses its chirality. (D) Doubling $\overline{\bfy}_0$ produces a looser spiral. All unspecified parameters in (B-D) are the same as in (A).}
\label{fig:lie_elements}
\vspace{-6mm}
\end{figure}
 
Indeed, various mathematical descriptions have been proposed for the shape and growth of shells, including parametrization of a shell surface \cite{raup1966geometric}, kinematic equations for a growing surface \cite{skalak1997kinematics}, and equations for a growth front that translates, rotates, scales, and twists  \cite{moulton2014surface}.   Our starting point is similar to \cite{moulton2014surface}, following the observation of \cite{thomson1917growth}.
%

\paragraph{Lie group representation}   Given that shell morphology arises from successive deposition at the growing edge, its evolution constitutes a continuous sequence of geometric transformations, each acting on the configuration of its predecessor.  This iterative structure aligns naturally with the composition and closure properties of a group, and specifically Lie groups \cite{helgason1979differential}. Lie groups are manifolds that are also groups.  They are essentially a set of rules that describe how to make continuous transformations in a consistent way. Associated with each such group is a Lie algebra\footnote{The Lie algebra is the tangent space of the Lie group at the identity element, equipped with the Lie bracket as its multiplication rule; see \cite{helgason1979differential}.}, which captures the linearized, infinitesimal motions that generate these transformations. In the case of molluscan shells, growth proceeds through a continuous sequence of dilations and rotations. Within this framework, the corresponding Lie algebra element, say $\bfA$, can be decomposed as two parts:
\beq
\bfA=\log\lambda\bfI+\bfW,\label{lie_algebra}
\eeq
where the scalar $\lambda$ characterizes the scaling\footnote{The scaling term is taken to be in logarithmic form to simplify the expressions in the calculations that follow.} and $\bfW=\bfw\times$
is a skew-symmetric tensor that characterizes the rotation about the axis $\bfw$. Applying the Lie algebra element (\ref{lie_algebra}) to the growth front $\bfy(t,x)$, parametrized by $x$ at time $t$, we obtain a partial differential equation governing its time evolution, 
\beq
\frac{\partial \bfy(t,x)}{\partial t}=\bfA\bfy(t,x),\quad x\in[L_1,L_2].\label{affine_y}
\eeq
By solving (\ref{affine_y}), the evolution of the growth front is given by 
\beq
\bfy(t,x)=\lambda^t\bfR(t)\bfy_0(x),\quad x\in[L_1,L_2],\label{seashell}
\eeq
where $\bfy(0,x)=\bfy_0(x)$ is the initial seed aperture and $\bfR(t)=\exp(t\bfW)$ is the rotation obtained via the exponential map of $\bfW$ (it is a rotation through angle $t$ and about a fixed rotational axis $\bfw$).  
Equation (\ref{seashell}) describes the shape of the shell: it depends on the elements $\lambda$ and $\bfW$ of the Lie algebra as well as the initial seed $\bfy_0$.  Referring back to Figure \ref{fig:lie_elements}A, the red line represents a line with constant $x$, while the sequence of black lines represents lines at different constant values of $t$.  

{ Figure \ref{fig:lie_elements} illustrates the Lie-group model (\ref{seashell}) and the roles of its parameters.  The shell in each panel is generated by rotating and scaling a circle.  The parameter $\lambda$ characterizes the expansion of the shell per unit growth increment (compare Fig.~\ref{fig:lie_elements}A and B), while $\bfw$, the rotation axis of $\bfR$, determines the direction of rotation (compare Fig.~\ref{fig:lie_elements}A and C). One can change the chirality of a spiral shell by reversing $\bfw$.  Note that the initial growth seed may be decomposed as $\bfy_0(x) = \widetilde \bfy_0 (x) + \overline \bfy_0$, where $\widetilde{\bfy}_0$ has its centroid at the origin and describes the shape of the generating curve, while $\overline{\bfy}_0$ is independent of $x$ and specifies its position relative to the rotation axis.  Thus,
\begin{equation}
\bfy(x,t) = \lambda^t\bfR(t)\widetilde \bfy_0(x) + \lambda^t\bfR(t) \overline \bfy_0.
\end{equation}
The first term describes the shape of the growth front, which remains unchanged up to scaling and rotation. The second term is independent of $x$ along the growth front and represents a time-dependent translation. The vector $\overline{\bfy}_0$ determines the distance of the generating curve from the rotation axis and, consequently, how tightly the shell coils (compare Fig.~\ref{fig:lie_elements}A and D).
}

We also note that  (\ref{seashell}) follows the orbit of a conformal Lie group (i.e., a group of continuous transformations that preserve the conformal geometry of the space) given by
\beq \label{group}
\calG=\{ \lambda^t\bfR(t), t \in {\mathbb R} \},
\eeq
with the group action $\calG(\bfy_0(x)) = \lambda^t\bfR(t)\bfy_0(x)$.

\paragraph{Initiation of growth}
The application of the Lie group to shell growth requires that we identify the initial growth seed $\bfy_0$.   It is known that molluscan eggs experience determinate spiral cleavage \cite{marin2007molluscan,huan2020dorsoventral} (except in Coleoidea), which establishes a helicospiral form of the protoconch during embryonic and larval development \cite{jackson2007dynamic}.   In homeostrophic molluscan shells, where the protoconch and teleoconch share a common axis, the edge of the protoconch defines the initial front $\widetilde \bfy_0$.  In heterostrophic shells, where their axes differ, there is a reconstruction of the mantle during the metamorphic stage, and this determines the initial front $ \widetilde \bfy_0$. { The translation $\overline \bfy_0$ is part of the growth law, and linked to the local growth law to be discussed later.}



\paragraph{Shell morphology, Lie group elements, and phylogenetic tree}
Figure \ref{fig:shells}(A)  shows the results of applying the ideas above to a wide variety of molluscan shell morphologies.  In each case, the shape that is obtained in Mathematica by applying (\ref{seashell}) with the elements listed in Table \ref{tab:param} is shown on the left, while the corresponding image of the actual shell is shown on the right.  The shapes of molluscan shells that we have studied, and these span a large variety, can be obtained in this manner.

\begin{figure}[t]
\centering
\includegraphics[width=\textwidth]{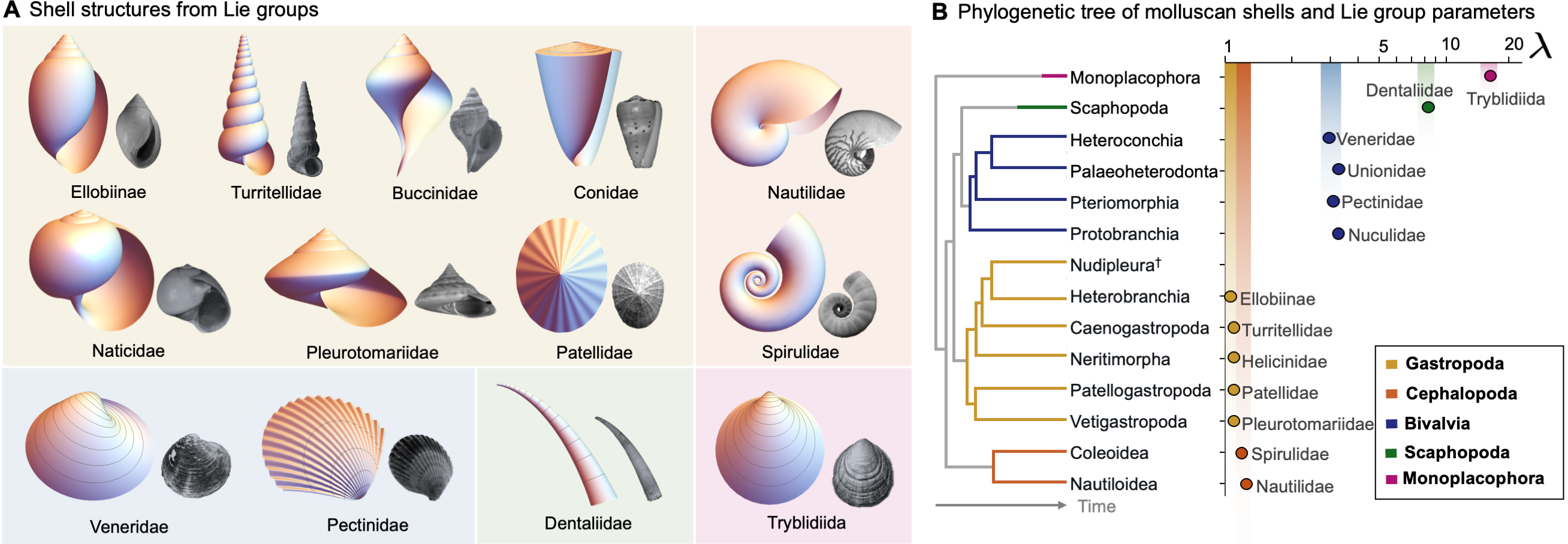}
\caption{(A). Representative molluscan shells obtained by Lie groups \eqref{group}. Shell families enclosed in boxes of the same color belong to the same class, and the grey subfigures display natural shell specimens. Different families of molluscan shells correspond to different Lie algebra parameters $\lambda$ and the seed aperture $\bfy_0$ (Table \ref{tab:param}). Licenses of gray subfigures: Buccinidae, Nautilidae, Spirulidae, and Tryblidiida are from Wikimedia Commons (public domain); Conidae is from an image by James St.~John via Wikimedia Commons, licensed under CC BY 2.0;  Patellidae is from Wikimedia Commons (James St.~John);  Pleurotomariidae is from an image by Shellnut via Wikimedia Commons, licensed under CC BY-SA 3.0; Dentaliidae is from an image from the Museum of New Zealand Te Papa Tongarewa, licensed under CC BY 4.0. The remaining gray subfigures were taken by the authors. (B) Phylogeny of mollusks based on \cite{chen2025genome} and the values of $\lambda$ used in this work to generate the corresponding classes;  Polyplacophora, Caudofoveata, and  Solenogastres are excluded as they either lack shells or possess them only rarely, as is  Nudipleura$^\dagger$ showed in the figure. 
}
\label{fig:shells}
\end{figure}

\setlength\tabcolsep{2pt}
\renewcommand{\arraystretch}{1.6}
\begin{table}
\caption{Geometric elements that describe the shape of molluscan shells shown in Figure \ref{fig:shells}.  \label{tab:param}}
\noindent\begin{tabular}{|c|c|c|l|c|c|c|c|c}
\hline
\textbf{Class} &
\textbf{Family} &
\bf{\makecell[c]{$\lambda$}}& 
\bf{\makecell[c]{$\bfy_0(x)$}}& t {\footnotesize  \rm (range)}
\\
\hline
\multirow{7}{*}{Gastropoda}&Turritellidae & 1.033  & {\footnotesize  $\makecell[l]{(0.719-0.965\cos(x - 0.2)+0.261\sin x) \hat\bfe_1 +\\ (6.927-0.26\cos(0.2-x)-0.965\sin x) \hat\bfe_2,\; x\in[0,2\pi]}$}&[0,70] \\
\cline{2-5}
&Ellobiinae & 1.06&  {\footnotesize $\makecell[l]{(0.992\sin x+0.238\cos x-0.556) \hat\bfe_1+\\ (-0.129\sin x+1.831\cos x+3.092) \hat\bfe_2+\\(0.004\sin x-0.11\cos x-0.182)\hat\bfe_3,\quad x\in[0,2\pi]}$}& [0,40]\\
\cline{2-5}
&Conidae & 1.04 & {\footnotesize  $\makecell[l]{(0.299\cos x +0.002x^2\sin x-0.024) \hat\bfe_1 +\\ (0.018\cos x-0.033x^2\sin x+0.059) \hat\bfe_2,\; x\in [1.6\pi,2.35\pi] }$}& [0,50] \\
\cline{2-5}
&Buccinidae &1.1 &{\footnotesize  $\makecell[l]{\tanh(x + 1.8) (0.096x^2-0.594x+0.92)\hat\bfe_1+\\  \tanh(x + 1.8) (0.029x^2-0.178x+0.276)\hat\bfe_2+\\ (-0.287x+0.307)\hat\bfe_1+(0.958x+2.806)\hat\bfe_2,x\in[-2.1,2]}$}& [0,30] \\
\cline{2-5}
&Naticidae &1.1 &  {\footnotesize $(\sin x -0.7)\hat\bfe_1 + (1.3 \cos x+1.65) \hat\bfe_2,\;x\in[0, 2\pi]$}&[0,30] \\
\cline{2-5}
&\makecell[c]{Pleurotoma-\\riidae} & 1.1& {\footnotesize $\makecell[l]{(0.992\hat\bfe_1-0.129\hat\bfe_2+0.004\hat\bfe_3)\cos x(|\sin x|^p+|\cos x|^p)^{\frac{-1}{p}}+ \\  (0.084\hat\bfe_1+0.643\hat\bfe_2-0.039\hat\bfe_3)\sin x(|\sin x|^{p}+|\cos x|^{p})^{\frac{-1}{p}}-\\0.824\hat\bfe_1+1.416\hat\bfe_2-0.081\hat\bfe_3,\;p=1.5+0.3\sin x, x\in[0,2\pi]}$}& [0,50]\\
\cline{2-5}
&Patellidae$^1$ &1.1 & {\footnotesize $\sin x \hat\bfe_1 + 1.25 \cos x \hat\bfe_2 + (0.03 \sin 20x-1.4) \hat\bfe_3,\; x\in[0,2\pi]$}&[0,40]  \\
\hline
\multirow{2}{*}{Cephalopoda}
& Spirulidae & 1.2&  {\footnotesize $(\sin x-2.1) \hat\bfe_1 + \cos x \hat\bfe_2+0.5\hat\bfe_3,\; x\in[0,2\pi]$}& [0,17]\\
\cline{2-5}
&Nautilidae & 1.24&{\footnotesize$\makecell[l]{(1.2-1.3 \cos x) \hat\bfe_1  +(0.08+0.3 \sin(0.958-2 x)) \hat\bfe_3\\ - \sin x \hat\bfe_2,\; x\in [0,2\pi] }$}& [0,10]\\
\hline
\multirow{3}{*}{Bivalvia}&Veneridae & 3.1&{\footnotesize$\makecell[l]{(0.302\sin x+0.762 \cos x-0.667) \hat\bfe_1+\\ (0.905\sin x-0.241\cos x+0.211) \hat\bfe_2+\\ (0.302\sin x-0.038\cos x+0.033)\hat\bfe_3,\; x\in[0,2\pi] }$}&[0,3]  \\
\cline{2-5}
&Cardiidae$^2$&2.718 &{\footnotesize $\makecell[l]{(0.12+x\sin^2(0.5y)) \hat\bfe_1 + (0.1-0.75x\sin y) \hat\bfe_2+\\  
 0.25(1-e^{-12x})\sin(0.5y)\tanh(4.5\cos x-3.15) \hat\bfe_3,\\ x\in[0.95\pi,2.05\pi]}$}&[0,3] \\
 \cline{2-5}
&Pectinidae &3 &{\footnotesize $\makecell[l]{(1.05 \sin x-0.7) \hat\bfe_1 + \cos x \hat\bfe_2 + 
 (0.035 \sin 37x-0.25) \hat\bfe_3,\\ x\in[0.07,1.2],\; y\in[0,\pi]}$}&[0,3.5] \\
\hline
Scaphopoda&Dentaliidae & 8& {\footnotesize $(\sin x-6) \hat\bfe_1 + \cos x \hat\bfe_2+10\hat\bfe_3,\; x\in[0,2\pi]$}& [0,1.2]\\
\hline
Monoplacophora&Tryblidiida& 16.12& {\footnotesize $(\sin x-0.8) \hat\bfe_1 + \cos x \hat\bfe_2+0.5\hat\bfe_3,\; x\in[0,2\pi]$}& [0,1.4]\\
\hline
\multicolumn{5}{|p{6.75in}|}{$\cdot$ Above, $\bfw$ is fixed as $(0,1,0)$, and $\hat\bfe_1=(1,0,0), \hat\bfe_2=(0,1,0)$, and $\hat\bfe_3=(0,0,1)$.
}\\
\multicolumn{5}{|p{6.75in}|}{$\cdot$ {The functions for $\bfy_0$ are empirical fits. }}\\
\hline
\multicolumn{5}{|p{6.75in}|}{$^1$For the family Patellidae, (\ref{seashell}) degenerates to $\bfy(t,x)=\lambda^t\bfy_0(x)$, so  $\bfw$ vanishes. Since $\lambda$ and $t$ are coupled, rescaling them could leave the configuration unchanged. We present one choice here. The configuration of the shell depends primarily on the position of $\bfy_0$.}\\
\multicolumn{5}{|p{6.75in}|}{$^2$This corresponds to the 3D bivalve shell shown in Figure \ref{spine}(B), which belongs to the family Cardiidae.
}\\
\hline
\end{tabular}
\end{table}

\begin{figure}[!ht]
\centering
\includegraphics[width=\textwidth]{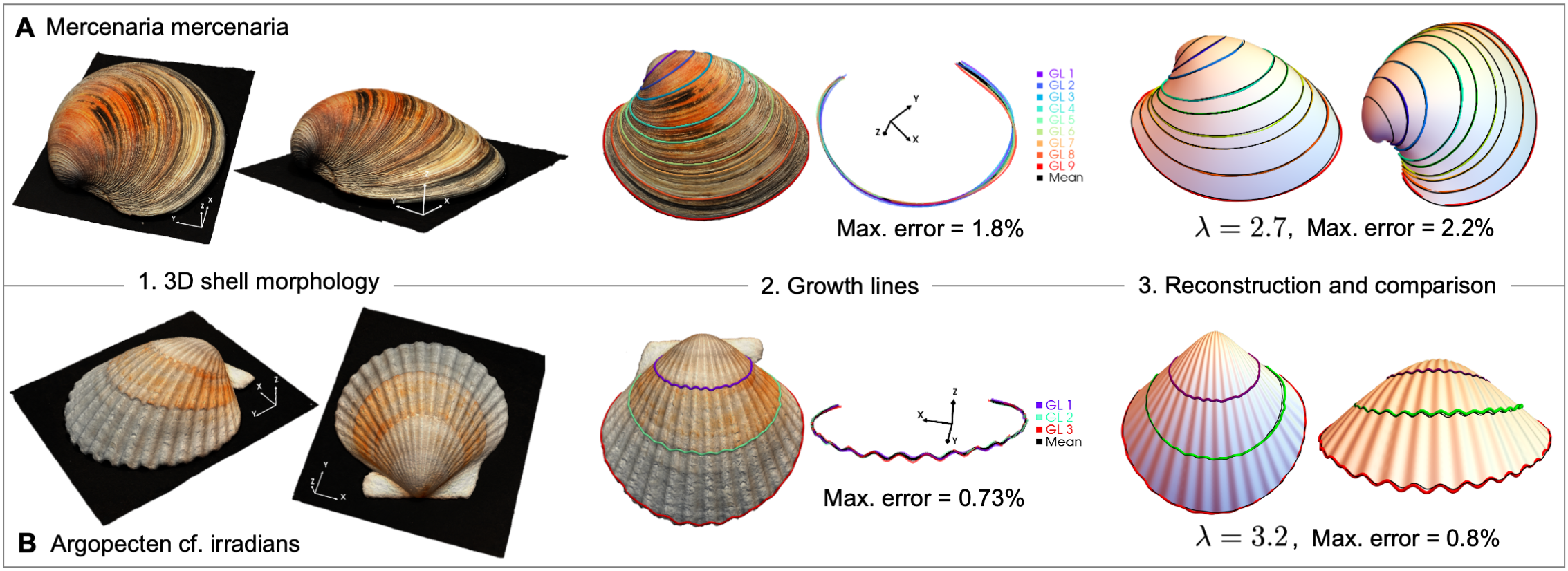}
\caption{ Quantitative comparison of theory with shells. (A) Clam (Mercenaria mercenaria) and (B) scallop (Argopecten cf.~irradians). From left to right: (1) 3D shell morphology obtained by photogrammetry; (2) extracted growth lines and their optimal alignment after removing rigid-body motions and scaling, showing that successive growth lines are self-similar with a maximum normalized error of 1.8\% (A) and 0.73\% (B); (3) shells reconstructed using the proposed Lie-group model with the mean growth line obtained from (2) and fitted scaling factors $\lambda=2.7$ (A) and $\lambda=3.2$ (B); and the comparison between observed (colored) and predicted (black) growth lines.
}
\label{fig:shell_comp}
\end{figure} 

{ A quantitative comparison between the theory and actual shells is shown in Figure \ref{fig:shell_comp} for a clam (Mercenaria mercenaria) and a scallop (Argopecten cf.~irradians). Since shell growth proceeds by material accretion at the shell margin, successive growth fronts are permanently preserved as growth lines. Thus, a mature shell records its developmental history, allowing the entire growth process to be reconstructed from a single specimen.  The shell morphology, including the growth lines, was reconstructed using photogrammetry (part 1). The extracted growth lines were translated, scaled, and rotated by best fit to reveal a common curve (part 2): the maximum error (maximum distance between two curves normalized by its length) is 1.8\% for the clam and 0.73\% for the scallop, thereby demonstrating that the growth lines are self-similar.  A scaled mean curve is taken as $\bfy_0$, and the Lie group is applied to it to reconstruct the entire shell (part 3); the resulting error (distance between the reconstructed and observed growth lines normalized by their length) is less than 2.2\% for the clam and 0.8\% for the scallop, demonstrating excellent agreement between the observed morphology and the theoretical prediction.  A third case involving a perturbation in $\bfy_0$ is studied later.}


Figure \ref{fig:shells}(B) shows a phylogenetic tree of (shell-forming) mollusks \cite{chen2025genome}, together with the Lie group scaling factor $\lambda$ used to reproduce the shell shape.  We observe that while $\lambda$ varies across molluscan groups, it takes very similar values within a class (with each class represented by a distinct color in the figure).  This evolutionary consistency of the parameters from the phylogenetic tree strengthens the conclusion that the growth mechanism is related to a Lie group structure.

\paragraph{Environmental conditions and reparameterization}

A remarkable aspect of molluscan shells is the robustness of shape.  Each species produces a shell with a characteristic shape independent of the environmental conditions { that can affect growth rates}, such as nutrition, pH, temperature, the concentration of ions, and growth rhythms.   Mathematically, the varying growth rates amount to a reparameterization of the group elements and thus of the shell surface.  To formalize this, we introduce a change of time variable $t\to f(t)$ to the shell surface given by (\ref{seashell}) to get $\bfy(t,x)=\lambda^{f(t)}\bfR(f(t))\bfy_0(x)$,
where $f(t)$ is a smooth function on $t$.  In this case, the new element of Lie algebra is  $ \bfA_{\rm new}=\dot f(t)\bfA$, which preserves the shell morphology. {Therefore, the resulting shell morphology is invariant to variations in growth rates
.} The time rescaling function $f(t)$ can be specified based on the actual temporal evolution of shell growth, like the Gompertz growth function \cite{hollyman2018age} and Von Bertalanffy's growth equation \cite{shelmerdine2007size}.  

\begin{figure}[t]
\centering
\includegraphics[width=4.5in]{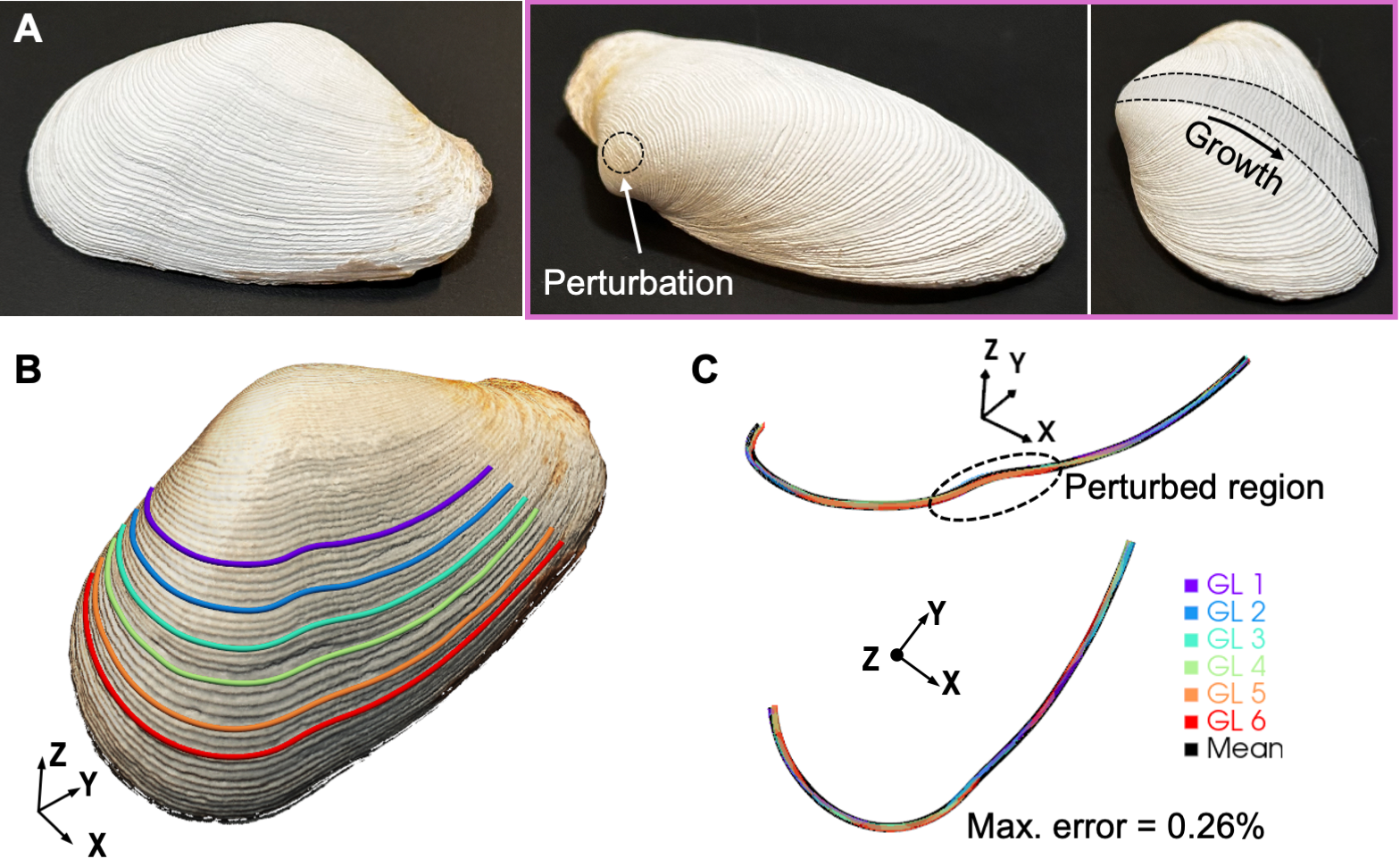}
\caption{Propagation of a perturbation during shell growth in a bent-nose clam (Macoma cf. nasuta). (A) Shell morphology showing a localized perturbation and its propagation during subsequent growth. (B) Growth lines extracted from the 3D morphology obtained by photogrammetry. (C) Growth lines after optimal alignment by scaling and rotation, showing that they remain self-similar while the perturbation propagates through subsequent growth under the Lie-group action.}
\label{fig:perturbation}
\end{figure} 

{
\paragraph{Perturbation to the growth front}
Injury or other damage can perturb the growth front at a particular time. Figure \ref{fig:perturbation} shows that such a perturbation propagates through subsequent growth in a bent-nose clam (Macoma cf. nasuta) without disrupting the underlying self-similar growth, with a maximum deviation of only 0.26\%. This demonstrates that the Lie-group structure is robust to perturbations of the growth front. Moreover, because the perturbation remains localized, the overall shell morphology is also robust to such perturbations.
}

\section{Local accretion laws}

\begin{figure}[t]
\centering
\includegraphics[width=0.73\textwidth]{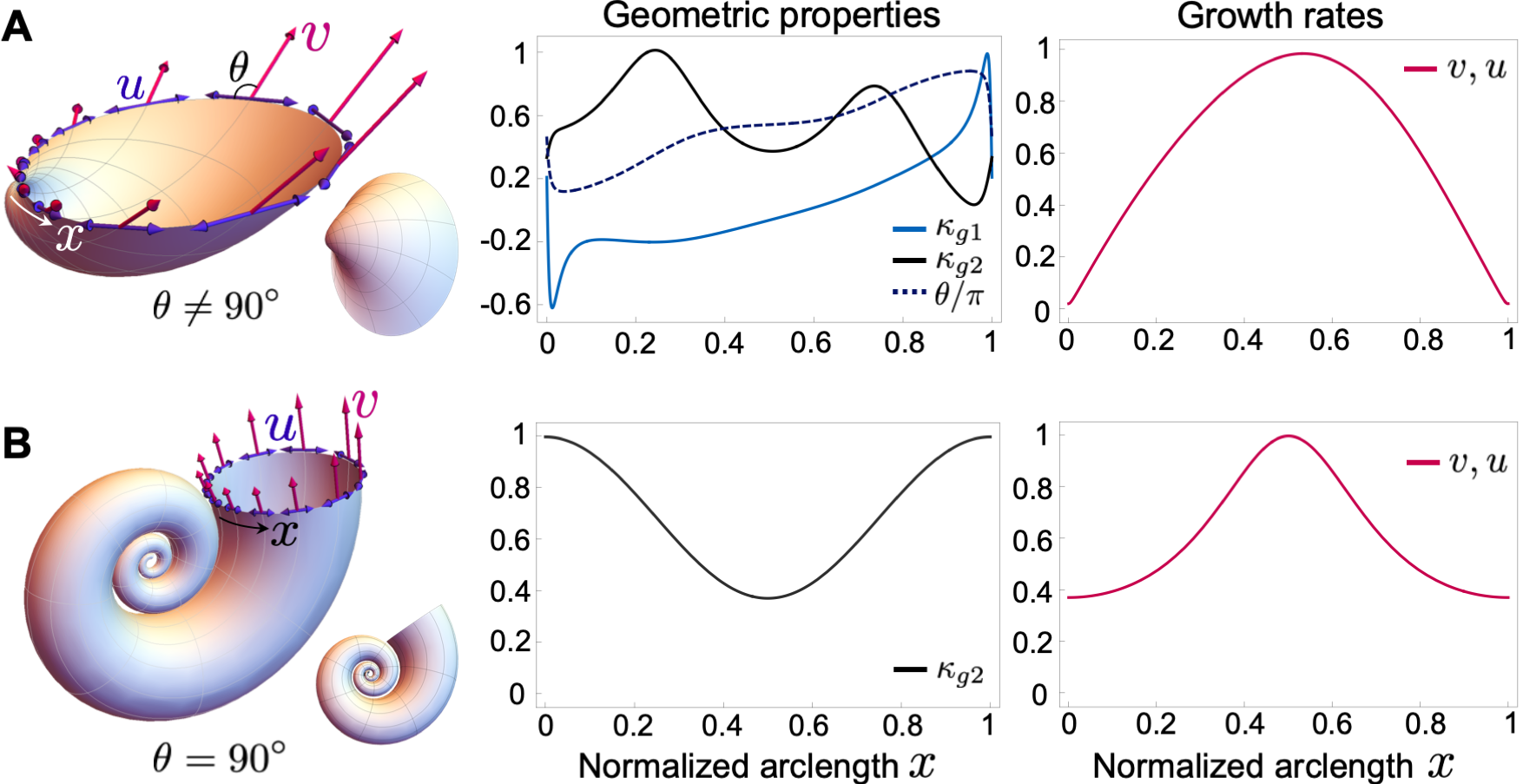}
\caption{Growth rates and local geometric information involved in the growth laws for (A) a bivalve shell, where the growth direction is not always perpendicular to the growth front ($\theta\neq90^\circ$), and (B) a nautilid shell with a perpendicular growth direction ($\theta=90^\circ$). Arrows on the left indicate the growth rates in the growth and lateral directions. The growth rates, geodesic curvatures, and total length of the growth front are normalized.}
\label{fig:law}
\end{figure}

The success of the Lie group framework described above concisely links the observation of the trajectory of the growing edge to the overall shape of the shells.  The further connection to phylogeny provides additional support for this framework.  However, it is still descriptive and leaves open the question of what biochemical processes can give rise to this evolution of the growing edge.

Within this framework, the shell admits a natural parametrization in terms of time $t$ and a parameter $x$ along the edge (not necessarily orthogonal to
$t$). This parametrization in turn defines two velocity components of the growing edge along these directions:
\begin{equation}
v=\left|\frac{\partial \bfy}{\partial t}\right|, \quad u=\left| \frac{\partial \bfy}{\partial x} \right|;
\end{equation}
$v$ is the growth rate along the growth direction (i.e., along the solid red lines of Figure \ref{fig:lie_elements}(A)), and $u$ is the lateral stretch of the growing edge (i.e., along the dashed black lines of Figure \ref{fig:lie_elements}(A)).  Note that $u, v$ depend on the parametrization, and in particular $v$ is not necessarily the parameterization-free normal velocity.  It turns out that the growth law (\ref{seashell}) is equivalent to the following relationship between the velocity components and the local curvature at the growing edge:
\beq
v(t,x)= u(t,x)=\frac{\log\lambda\sin\theta(x)+\theta_x(x)}{\kappa_{g2}(t,x)-\kappa_{g1}(t,x)\cos\theta(x)} \ .\label{eq:acc}
\eeq
Here, $\theta(x)\in(0,\pi)$ is the angle between the parameter curves at the growth front and it is independent of $t$ since the surface defined by (\ref{seashell}) is conformal, and $\theta_x$ denotes its derivative with respect to $x$; $\kappa_{g1}$ and $\kappa_{g2}$ are the geodesic curvatures of the growth trajectory at fixed $x$ and the growth edge at fixed $t$, respectively.  The geodesic curvature is the curvature of the orthogonal projection of the curve onto the tangent plane of the surface at $\bfy(t,x)$.  Figure \ref{fig:law} illustrates the growth rates and the local geometric information in the cases of $\theta\neq90^\circ$ and $\theta=90^\circ$, respectively. Specially, when $\theta=90^\circ$, the growth laws simplifies to
\beq
v(t,x)= u(t,x)=\frac{\log\lambda}{\kappa_{g2}(t,x)} \ , \label{eq:acc90}
\eeq
and depends only on $\kappa_{g2}$. We refer the reader to Appendix \ref{app:geom} for geometric preliminaries, 
and Theorem \ref{theorem1} and Remark \ref{remark1} in Appendix \ref{app:thm}  for a precise statement.

{ 
\paragraph{Curvature and self-assembly}
The exact mechanisms that lead to such a local growth law are not known.}  Curvature-regulated growth is common in both physical and natural systems. In physical systems,  curvature controls crystal growth through the Gibbs-Thomson effect \cite{langer,kammer2006morphological,asta2009solidification} and governs processes such as meandering river evolution \cite{finotello2024vegetation}.  Examples in natural systems include vascular morphogenesis \cite{luciano2024multiscale,bade2017curvature,mandrycky20203d} and tissue growth \cite{nelson2009geometric}.  Curvature reshapes the local physical or chemical environment, such as altering stress distributions, cell arrangements, interfacial biochemical gradients, { electrostatic fields}, or transport processes. These curvature-induced changes, in turn, modulate the underlying growth, migration, or erosion processes, whereby curvature emerges.

{  It is known that the mantle drives processes involving the assembly of nanoparticles, leading to the growth of the shell \cite{addadi2006mollusk,hovden2015nanoscale}. In other words, growth is mediated by the assembly of nanoparticles with defined shapes and sizes. It is also known that the self-assembly of nanoparticles can lead to, and be influenced by, curvature, including geodesic curvature \cite{laradji2026ordered,santangelo2007geometric}. Therefore, it is plausible that the local growth law involves geodesic curvature. We now provide a heuristic argument in the special case of orthogonal parametrization, where the growth law simplifies to (\ref{eq:acc90}). Recall that $u$ is the velocity along the edge, while $v$ is the velocity normal to the edge. The addition of a nanoparticle of defined shape and size will lead to proportional growth along and normal to the edge. Thus, up to parametrization, we expect $u=v$, consistent with (\ref{eq:acc90}). Now, since each nanoparticle is extremely small compared to the shell, we may consider the growth process to occur on the tangent plane. The geodesic curvature $\kappa_{g2}$ is the curvature of the growth line on this plane. We expect the addition of nanoparticles of defined shape and size to be easier in straighter regions of the growth front than in more highly curved regions, leading to an inverse dependence of the velocities on geodesic curvature, consistent with (\ref{eq:acc90}). This heuristic argument shows how the local growth law can arise from the process of nanoparticle assembly. The exact details, as well as the generalization to non-orthogonal parametrization, remain to be understood.
}

\paragraph{Environmental conditions and reparameterization}
We noted earlier that the shell shape remains invariant in a species even though the growth rates can depend on environmental conditions.  We formalized this by introducing a rescaling $f(t)$ and showed that the Lie group structure remains unchanged by the rescaling.  The same thing is true for the accretion laws (\ref{eq:acc}).  For example, the Gompertz law \cite{hollyman2018age}, $f(t)=-c\lambda^{-t},$ where $c>0$ is a growth coefficient associated with the mollusk, leads to the relationship 
\beq
v=\frac{c}{\lambda^t}\frac{\log\lambda\sin\theta+\theta_x}{\kappa_{g2}-\kappa_{g1}\cos\theta},\label{size_change}
\eeq
which preserves a relationship between the accretion velocity and the local curvature, up to a rescaling of time.  A similar transformation is true for  Von Bertalanffy's growth equation \cite{shelmerdine2007size}.

\paragraph{Perturbation to the growth front}
The result in Figure \ref{fig:perturbation} shows that the local growth law, and thus the underlying growth mechanism, remains unchanged in the bent-nose clam (Macoma cf.~nasuta) under perturbations to the growth front.

%
%
%

\section{Discussion}

The results above show that molluscan shells can be described in a concise mathematical manner as conformal Lie groups characterized by a scaling factor, a rotational axis, and a seed curve.  The scaling factor has a close relation with the phylogenetic tree of mollusks.  Further, we have shown that the Lie-group description of shell growth is equivalent to biologically passive, curvature-driven accretion at the shell edge. { We have provided a heuristic argument linking the local growth law to the self-assembly of nanoparticles.  We show that shell morphology remains invariant under environmental variations that alter growth rates, providing a potential mechanism for the developmental stability of shell form despite environmental variability, while the overall shape is also robust to perturbations in the growth line. Together, these results suggest that once biologically programmed through the selection of Lie-group parameters, shell growth can proceed according to local growth laws, without requiring active shape control by the mantle.
}


{The proposed framework differs from existing approaches in several important respects. Kinematic models typically describe shell form either through direct geometric parameterizations of the global morphology (e.g., \cite{thomson1917growth,contreras2023mathematical}) or by solving local curve-evolution equations governed by prescribed growth velocities (e.g., \cite{moulton2014surface}). In contrast, in the proposed Lie-group framework,  the global shell shape emerges directly from the underlying growth mechanism, establishing an explicit link between local growth laws and macroscopic morphology.  More broadly, recent studies have demonstrated how geometric principles can illuminate the relationships among biological form, development, function, and evolution \cite{al2021geometry}. In contrast to this approach that focuses on identifying geometric descriptors or growth laws for a particular biological system, the present work develops a common group-theoretic framework for shell morphogenesis, and a wide range of shell characteristics arise from a common mathematical structure. Within the same Lie-group formulation, symmetry, chirality, self-similar growth, phylogenetic scaling relations, growth-line patterns, shell thickening, and ornamentation emerge naturally as consequences of the underlying group action. This provides a unified geometric description linking local growth kinematics to diverse aspects of shell morphology.}

Beyond molluscan shells, such morphology is widespread in nature, ranging from animal tusks, horns, beaks to insect spines and mouthparts, plant tips and thorns, and even microbial shells and skeletons. Their growth dynamics may also be accounted for by the proposed laws.

\paragraph{Spines and fronds}  
\begin{figure}[t]
\centering
\includegraphics[width=0.65\textwidth]{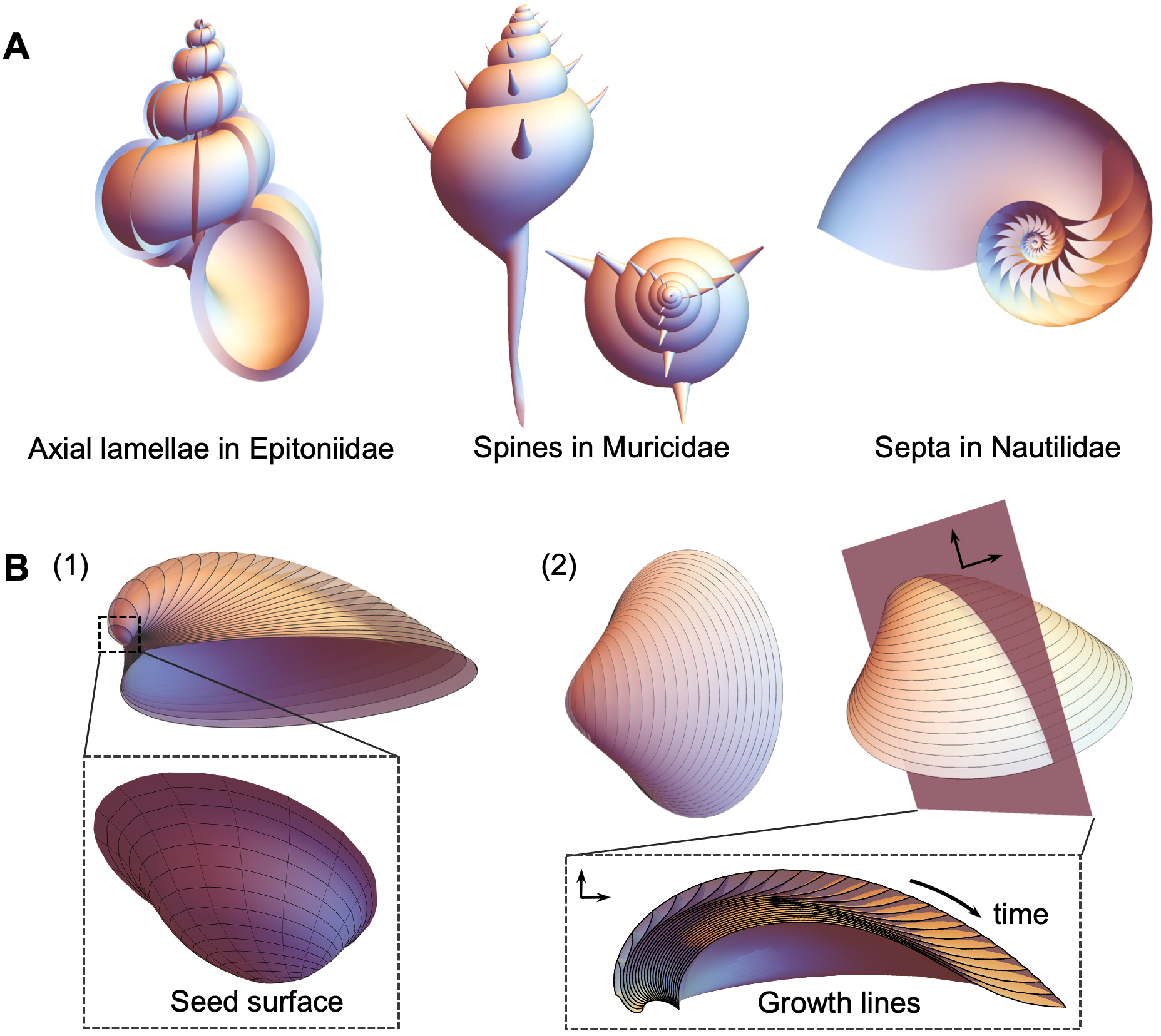}
\caption{A. Representative shell ornamentation generated by discrete subgroups of the Lie group used to generate the main shell surface.  B. Representative bivalve shell obtained by the Lie group, where (1) shows the seed surface, and (2) shows the growth lines on the shell's outer surface and in the cross-section during evolution. The corresponding Lie group parameters and the seed surface are summarized in Table \ref{tab:param}, and a video illustrating the growth process is provided in the Supplemental Information.}
\label{spine}
\end{figure} 

Some molluscan shells, such as those of Murex and Epitonium scalare, display highly complex forms sculptured with spines or fronds. This process involves temporal changes in gene expression. We hypothesize that variations in gene expression could lead to the emergence of distinct governing growth laws. In this case, each spine or frond could follow a Lie group orbit distinct from that of the main shell surface, producing a unique sculpted form. More interestingly, the distribution of spines, fronds, and respiratory tremata (as observed in the Haliotis shell) typically follows a specific orbit of a subgroup of the primary Lie group, as illustrated in Figure \ref{spine}(A). These patterns underscore the central role of group structure in molluscan shell growth.

\paragraph{Three dimensions}  Strikingly, the Lie-group representation that governs surface morphology also extends through the whole shell body, regulating its formation across the entire thickness. Bivalves are a canonical example whose growth lines across the thickness are well documented \cite{jones1996marking,jones1983sclerochronology}. A typical bivalve shell comprises three layers: from the outside in, the periostracum, the prismatic layer, and the nacreous layer; the latter two are mineralized. Compared to the nacreous layer with finer growth lines, the prismatic layer displays more widely spaced lines as a result of its rapid deposition. Noting that the evolution of these growth lines adheres to the local growth mechanisms described above, it follows that the entire shell body can be described using Lie groups (\ref{group}). In this case, the shell body is obtained by replacing the seed curve $\bfy_0(x), x\in[L_1,L_2]$ in (\ref{seashell}) with a seed surface $\bfy_0(x,y),x\in[L_1,L_2],y\in[L_3,L_4]$. As illustrated in Figure \ref{spine}(B), the bivalve shell is generated by applying the Lie group to the seed surface in B(1), yielding the growth lines in B(2) that closely match those of a real shell. 


\paragraph{Mechanical forces}    Soft-bodied molluscs generally have a bilaterally symmetric body, yet their shells often adopt a helicospiral form. \cite{chirat2021physical} hypothesized that a growth mismatch between the shell tube and the symmetric body generates mechanical forces, which are relieved by a body twist, resulting in shell asymmetry. However, this hypothesis can not explain why evolution does not favor a helicospiral body that matches the snail shell, or why the shell form of the same species remains so consistent across different individuals under varying soft-body growth conditions.  Here our growth laws (\ref{eq:acc}) indicate that shell growth is fully determined by a single scaling parameter, together with regulation via sensing the intrinsic geodesic curvature of the growth front. 
The shell's twisting feature then emerges naturally from these laws when applied to a protoconch, whose twist is induced by spiral cleavage at an early stage of development.

\paragraph{Role of group structure} During growth, organisms typically exhibit geometric similarity in morphology and dynamic similarity in motion \cite{liu2025comprehensive}. This geometric similarity indicates that the organism often grows according to conformal Lie groups; such growth may result from the evolutionary processes in which the organism retains an efficient dynamic performance under environmental pressure as its body size increases. This Lie group representation provides an opportunity to explore growth laws for general organisms using the same framework developed here. Moreover, for three-dimensional bulk objects, the only admissible conformal mapping is uniform scaling in all dimensions. Different from the surface of molluscan shells, which also involves rotation, this reduced complexity may facilitate the identification of the growth laws in the three-dimensional bodies of organisms.


\paragraph{Outlook}This framework holds potential for both basic and applied sciences. In biomimetic engineering, it offers a computational grammar for designing materials and structures that replicate the growth and functionality of real organisms; It also motivates an evolution-inspired design paradigm, in which structures can be grown to adapt to environmental stimuli, giving rise to novel forms of soft robotics. In manufacturing, these local growth laws can be integrated with physical or chemical reaction systems to grow targeted structures \cite{kim2025morphogenic}, facilitating sustainable production through substantial energy savings. In pathology, the ability to model the geometric progression of forms could offer new insights into aberrant growth patterns, potentially informing predictive diagnostics and targeted therapies. Future research could explore their applicability to other classes of biological structures and uncover new growth principles for general forms. Another important direction lies in integrating these fundamental laws with empirical growth laws from chemical and physical systems, aiming to advance novel manufacturing methods.

\paragraph{Acknowledgement.} We thank Professors Rob B.~Phillips and Lea Goentoro for insightful discussions.  H.L. acknowledges the support of the Drinkward Postdoctoral Fellowship and K.B. acknowledges the support of Howell N.\ Tyson, Sr., Professorship at Caltech.

\paragraph{Competing interest.} The authors declare no competing interests.

\paragraph{Data sharing.} All data required for this work are incorporated in the manuscript.


\newpage
\appendix
\section*{Appendix}
\renewcommand{\theequation}{A\arabic{equation}}
\renewcommand{\thefigure}{A\arabic{figure}}
\setcounter{equation}{0}
\setcounter{figure}{0}


\section{Geometric preliminaries} \label{app:geom}

We recall some basic properties of manifolds (see, for example, \cite{kuhnel2015differential,kreyszig1968introduction}).  Consider a  $C^3$-smooth surface ( two-dimensional manifold embedded in three-dimensional Euclidean space) parametrized by parameters $\alpha^1, \alpha^2$ and described as  $\bfy(\alpha^1, \alpha^2)$.  Define the natural basis of the tangent space as $\bfe_i = \bfy_{,i}, \ i = 1, 2$ where the subscript $,i$ denotes partial differentiation with respect to $\alpha_i$, and the reciprocal basis of the tangent space as $\bfe^i$ through the relation $\bfe _i \cdot \bfe^j = \delta_i^j$.  The unit normal to the manifold is $\hat \bfn = (\bfe_1 \times \bfe_2)/|\bfe_1 \times \bfe_2|$, and the projection that maps a three-dimensional vector to the tangent space is
\beq
\bfP = \bfI - \hat \bfn \otimes \hat \bfn = \sum{\bfe_i \otimes \bfe^i}\label{P}.
\eeq

The components of the metric tensor (in the two bases), the second fundamental form, and the Christoffel symbols are defined as
\beq
g_{ij}:=\bfe_i\cdot\bfe_j, \quad g^{ij}:=\bfe^i\cdot\bfe^j, \quad b_{ij}:=\bfy_{,ij}\cdot\bfn,  \quad {{\Gamma_i}^k}_j:=g^{km}\bfy_{,ij}\cdot\bfe_m,\ i,j,k,m=1,2,
\eeq
respectively.  Now, the smoothness of the surface implies the compatibility conditions (equality of mixed derivatives)
\beq
\bfy_{,12}=\bfy_{,21},\quad \bfy_{,112}=\bfy_{,121},\quad \bfy_{,221}=\bfy_{,212}. \label{compatibility}
\eeq
We can use the definition of the natural basis, the second fundamental form, and the Christoffel symbols to show $\bfy_{,ij}={{\Gamma_i}^k}_j\bfe_k+b_{ij}\bfn$, and therefore the compatibility conditions can be represented as  
\beq
\left\{\begin{array}{rcl}
b_{12}&=&b_{21},\\
{{\Gamma_1}^1}_2&=&{{\Gamma_2}^1}_1,\\
{{\Gamma_1}^2}_2&=&{{\Gamma_2}^2}_1,\\
b_{11,2}-b_{12,1}&=&b_{11}{{\Gamma_1}^1}_2+b_{12}({{\Gamma_1}^2}_2-{{\Gamma_1}^1}_1)-b_{22}{{\Gamma_1}^2}_1,\\
b_{12,2}-b_{22,1}&=&b_{11}{{\Gamma_2}^1}_2+b_{12}({{\Gamma_2}^2}_2-{{\Gamma_1}^1}_2)-b_{22}{{\Gamma_1}^2}_2,\\
g^{hi}(b_{jl}b_{kh}-b_{jk}b_{lh})&=&({{\Gamma_j}^i}_l)_{,k}-({{\Gamma_j}^i}_k)_{,l}+{{\Gamma_j}^p}_l{{\Gamma_p}^i}_k-{{\Gamma_j}^p}_k{{\Gamma_p}^i}_l,\\
&&\hfill{\rm with}\; i=1,2, \quad jlk=112,221.  
\end{array}\right.\label{GW-GC}
\eeq
The equations (\ref{GW-GC})$_{4,5}$ are the Codazzi-Mainardi formulae, and the four relations in (\ref{GW-GC})$_{6}$ are known as the Gauss equations. 

We now specialize to the setting in which the shell (surface) evolves over time $t$.  We parametrize the shell as $\bfy(t,x)$, where $x$ parametrizes the edge of the shell at any fixed time.   Here we identify $\alpha^1 = t, \alpha^2 = x$. Let $\theta$ be the angle between the two coordinate curves.  It follows 
\beq
\bfe_1=\frac{\partial \bfy}{\partial t}=\bfv=v \hat{\bfe}_1,\quad
\bfe_2=\frac{\partial \bfy}{\partial x}=\bfu=u\hat{\bfe}_2, \quad 
\bfe^1=-\frac{1}{v\sin\theta} \hat \bfe_1^\perp,\quad
\bfe^2=\frac{1}{u\sin\theta} \hat \bfe_2^\perp.  
\eeq
Thus, the first natural basis vector ($\bfe_1$) is identified as the real growth velocity, and the second ($\bfe_2$) is identified as the lateral dilatation of the growing edge of the shell.  Also, the angle $\theta$ between the natural basis vectors coincides with the angle that the growth velocity $\bfv$ makes with the edge.   Now, the triads, $\{\hat{\bfe}_1, \hat{\bfe}_1^\perp, \hat{\bfn}\}$ and $\{\hat{\bfe}_2, \hat{\bfe}_2^\perp, \hat{\bfn}\}$ are two Darboux frames (orthonormal basis in three-dimensional Euclidean space with one vector tangent to a curve and one vector normal to the surface) along the $t$ and $x$ curves, respectively.  We define $(\kappa_{g1},\kappa_{n1},\tau_{r1})$ and $(\kappa_{g2},\kappa_{n2},\tau_{r2})$ to be geodesic curvature, normal curvature, and geodesic torsion of curves at fixed $t$ (black curve in Figure \ref{fig:lie_elements}(A)) and $x$ (red curve in Figure \ref{fig:lie_elements}(A)) respectively through
\beq
\left\{\begin{array}{lcl}
(\hat \bfe_1)_{,t}&=& v(\kappa_{g1}\hat \bfe_1^\perp+\kappa_{n1}\bfn),\\
(\hat \bfe_1^\perp)_{,t}&=&v(-\kappa_{g1}\hat \bfe_1+\tau_{r1}\bfn),\\
\bfn_{,t}&=&v(-\kappa_{n1}\hat \bfe_1-\tau_{r1}\hat \bfe_1^\perp),
\end{array}\right.\quad{\rm and}\quad
\left\{\begin{array}{lcl}
(\hat \bfe_2)_{,x}&=& u(\kappa_{g2}\hat \bfe_2^\perp+\kappa_{n2}\bfn),\\
(\hat \bfe_2^\perp)_{,x}&=&u(-\kappa_{g2}\hat \bfe_2+\tau_{r2}\bfn),\\
\bfn_{,x}&=&u(-\kappa_{n2}\hat \bfe_2-\tau_{r2}\hat \bfe_2^\perp).
\end{array}\right.
\label{darboux}
\eeq
We now have an explicit representation of the metric components as 
\beq
\quad g_{ij}=\bfe_i\cdot\bfe_j=\left[\begin{array}{cc}
     v^2&uv\cos\theta  \\
uv\cos\theta & u^2
\end{array}\right],\quad g^{ij}=\bfe^i\cdot\bfe^j=\left[\begin{array}{cc}
     \frac{1}{v^2\sin^2\theta}&\frac{-\cos\theta}{uv\sin^2\theta} \\
    \frac{-\cos\theta}{uv\sin^2\theta}  & \frac{1}{u^2\sin^2\theta}
\end{array}\right].
\eeq
Assume that $\theta$ is independent of $t$. Then the second fundamental form, as well as the Christoffel symbols, are found by
\beqs
&&b_{ij}=\left[\begin{array}{cc}
    v^2\kappa_{n1} &vu(\kappa_{n1}\cos\theta+\tau_{r1}\sin\theta)  \\
    vu(\kappa_{n2}\cos\theta-\tau_{r2}\sin\theta) & u^2\kappa_{n2}
\end{array}\right],\nonumber\\
&&\begin{aligned}
{{\Gamma_1}^1}_1&=\frac{v_t}{v}-\frac{v\kappa_{g1}\cos\theta}{\sin\theta},&{{\Gamma_1}^2}_1&=\frac{v^2\kappa_{g1}}{u\sin\theta},\\
{{\Gamma_1}^1}_2&=\frac{v_x}{v}+\frac{(\theta_x-u\kappa_{g2})\cos\theta}{\sin\theta},&{{\Gamma_1}^2}_2&=\frac{v\kappa_{g2}}{\sin\theta}-\frac{v\theta_x}{u\sin\theta},\\
{{\Gamma_2}^1}_1&=-\frac{u\kappa_{g1}}{\sin\theta},&{{\Gamma_2}^2}_1&=\frac{u_t}{u}+\frac{v\kappa_{g1}\cos\theta}{\sin\theta},\\
{{\Gamma_2}^1}_2&=-\frac{u_x}{v^2\sin\theta},&{{\Gamma_2}^2}_2&=\frac{u_x}{uv}\frac{\cos\theta}{\sin\theta}-u\kappa_{g2}.\label{chris_b}
\end{aligned}
\eeqs

Finally, we recall the definition of a covariant derivative of a vector field that describes the tangential component of the derivative of a tangential vector $\bfz = \hat z^i \bfe_i $ along a curve on the manifold,
\beq
\nabla \bfz = \left(\frac{\partial \hat z^i}{\partial x^j}+\hat z^k{{\Gamma_k}^i}_j \right)\bfe_i\otimes\bfe^j.\label{covariant_def}
\eeq
Note that it is the projection of the directional derivative onto the tangent plane of the surface.  Applying (\ref{covariant_def}) to the growth velocity $\bfv$ that coincides with the first natural vector, we find
\beq \label{eq:covv}
\nabla \bfv = \nabla \bfe_1 = {{\Gamma_1}^i}_j \bfe_i\otimes\bfe^j.
\eeq

\section{Relationship between global and local growth laws} \label{app:thm}

We now show that the global growth determined by the Lie algebra, $\bfy(t,x) = \lambda^t \bfR(t) \bfy_0(x)$, is equivalent to a growth law (\ref{growth_laws}) at the growing edge that depends only on the local geometry.  We have the following theorem.

\begin{theorem}
\label{theorem1}
Let $\mathcal{S}_t\subset C^3([0,t]\times[L_1,L_2])$ be a time-evolving manifold, whose growth front at $t$ is defined by 
\beq
\bfy(t,x)=\lambda^t\bfR(t)\bfy_0(x), \quad x\in [L_1,L_2],\; t\geq 0,\label{manifold}
\eeq
where $\lambda>0$ with $\lambda\neq1$; $\bfR(t)$, with $\bfR(0)=\bfI$, is a rotation through angle $t$ about a fixed axis $\bfw$; and $\bfy_0(x)=\bfy(0,x),x\in [L_1,L_2]$ is a closed seed curve satisfying $\bfy_{0,xx}(x)\cdot\bfn_0\neq0$; here $\bfn_0=\bfn(0,x)$ denotes the surface normal at $t=0$. Define the growth velocity and the lateral dilatation as $ \bfv(t,x) = \partial\bfy(t,x)/\partial t$ and $\bfu(t,x)=\partial\bfy(t,x)/\partial x$ and denote their magnitudes by $v$ and $u$, respectively.

Then a necessary and sufficient condition for the manifold (\ref{manifold}) to grow uniquely is that the magnitudes $v$ and $u$ satisfy 
\beq
u(t,x)=v(t,x)=\frac{\log\lambda\sin\theta(x)+\theta_{,x}(x)}{\kappa_{g2}(t,x)-\kappa_{g1}(t,x)\cos\theta(x)},\label{growth_laws}
\eeq
with the initial conditions of $\bfv$ and $\bfu$
\beq
\bfv_0=(\log\lambda\bfI+\bfW)\bfy_0(x), \quad \bfu_0=\bfy_{0,x}(x),\label{initial_vu}
\eeq
where $\theta$ is the angle between the two growth directions $\bfv$ and $\bfu$, independent of $t$; $\bfW=\bfw\times$; and $\kappa_{g1}(t,x)$ and $\kappa_{g2}(t,x)$ are the geodesic curvatures at $\bfy(t,x)$ with $x$ and $t$ held fixed, respectively.
\end{theorem}

\vspace{\baselineskip}
\begin{proof}
We begin by showing that the local growth law \eqref{growth_laws} is {\it necessary} for the global growth law \eqref{manifold}.  By differentiating \eqref{manifold} with time, we observe that $\bfv = \bfA \bfy$ with $\bfA = \log \lambda \bfI + \bfW$.   Therefore, the covariant derivative of $\bfv$ is 
\beq
\nabla\bfv=\bfP \, \frac{\partial \, (\bfA\bfy)}{\partial \alpha^i}\otimes\bfe^i=\bfP\bfA\frac{\partial\bfy}{\partial \alpha^i}\otimes\bfe^i=\bfP\bfA\bfe_i\otimes\bfe^i=\bfP\bfA\bfP. \label{governing}
\eeq
Combining with \eqref{eq:covv}, we conclude that 
\beq
{{\Gamma_1}^i}_j = \bfe^i \cdot (\nabla \bfv) \bfe_j = \log \lambda \delta^i_j + \bfe^i \cdot \bfW \bfe_j 
=  \log \lambda \delta^i_j +  \bfw \cdot ( \bfe_j \times \bfe^i ).\label{covariant_christoffel}
\eeq
Recalling the explicit representation of the Christoffel symbols from \eqref{chris_b}, we obtain after a long but straightforward calculation the following relations:
\beq
\frac{v_{,t}}{v}=\log\lambda,\quad\frac{u_{,t}}{u}=\log\lambda,\quad\frac{u_{,t}\cos\theta-v_{,x}}{u\sin\theta}=v\kappa_{g1}=\bfw\cdot\bfn,\label{S20}
\eeq
where the first equation in (\ref{S20})$_3$ can also be obtained by the compatibility conditions (\ref{GW-GC})$_{2,3}$.

We now show that  $\theta $ is independent of $t$ for the resulting parameterization of the surface.  Note that $ \cos \theta = \frac{\bfv \cdot \bfu}{|\bfv||\bfu|} = \frac{\bfA \bfy \cdot \lambda^t \bfR \bfy_{0,x}}{|\bfA \bfy ||\lambda^t \bfR \bfy_{0,x}|} = \frac{\bfA \bfy_0 \cdot \bfy_{0,x}}{|\bfA \bfy_0  ||\bfy_{0,x}|}=\frac{\bfv_0\cdot\bfu_0}{|\bfv_0||\bfu_0|}$.  Thus, we conclude that $\theta$ only depends on $x$ and is determined by the initial condition.

We then proceed to establish \eqref{growth_laws}.  Note that (\ref{S20})$_{1,2}$ implies that $\partial  (\log u - \log v)/\partial t  = 0$,  Consequently, $u=v g(x)$, where $g(x)$ depends only on $x$, and can be determined from the initial condition: $g(x)=u(0,x)/v(0,x)$.  By parameterizing the initial curve $\bfy_0(x)$ such that $|\bfy_{0,x}(x)|=v(0,x)$, $g(x)$ is simplified to $1$, and therefore
\beq \label{eq:uv}
v(t,x) =u(t,x).
\eeq
Combining this with the compatibility condition (\ref{GW-GC})$_3$, we conclude
\beq
v=\frac{\log\lambda\sin\theta}{\kappa_{g2}-\kappa_{g1}\cos\theta-\theta_{,x}/u}.\label{condition20}
\eeq
Using (\ref{eq:uv}) again, we obtain (\ref{growth_laws}), establishing the necessity.
\vspace{\baselineskip}

We now turn to proving that the local growth laws \eqref{growth_laws}, together with the initial conditions \eqref{initial_vu}, are {\it sufficient} for the global growth law \eqref{manifold}.  Note that \eqref{growth_laws} implies \eqref{condition20}.   Combining this with  the compatibility condition $u_{,t}\sin\theta=uv(\kappa_{g2}-\kappa_{g1}\cos\theta)-v\theta_{,x}$ (cf., (\ref{GW-GC})$_3$), we obtain $u_{,t}/u= v_{,t}/v = \log\lambda$.   Thus,
\beq
v(t,x)=\lambda^t v(0,x),\quad u(t,x)=\lambda^t u(0,x).\label{speeds}
\eeq
Since we assume that the angle between the growth direction and edge $\theta$ is independent of time, there exists a rotation $\bfR(t,x) \in \mathrm{SO}(3)$ with  the initial value $\bfR(0,x)=\bfI$ such that
\beq
\left\{ \frac{\bfv(t,x)}{v(t,x)} ,  \frac{\bfu(t,x)}{u(t,x)},\bfn(t,x) \right\} = \bfR(t,x)\  \left\{ \frac{\bfv_0}{v(0,x)},  \frac{\bfu_0}{u(0,x)},\bfn_0 \right\}.\label{frame1}
\eeq
Combining the initial condition (\ref{initial_vu}) with (\ref{speeds})-(\ref{frame1}), we have the following expressions,
\beqs
\bfv(t,x)&=&\lambda^t\bfR(t,x)(\log\lambda\bfI+\bfW)\bfy_0(x),\nonumber\\
\bfu(t,x)&=&\lambda^t\bfR(t,x)\bfy_0'(x).\label{velo_v}
\eeqs

We now show that $\bfR_{,t} = \bfR \bfW$ for all $t$ where $\bfW$ is given by the initial conditions \eqref{initial_vu}.  To that end, define skew $\bfW_1=\bfw_1\times$ and $\bfW_2=\bfw_2\times$ by $\bfW_1=\bfR^{\rm T}\bfR_{,t}$ and $\bfW_2=\bfR^{\rm T}\bfR_{,x}$.  We seek to show $\bfW_1 = \bfW, \bfW_2 = \bf0$.  For any $t$, we have the compatibility conditions \eqref{compatibility}, which imply 
\beq
\bfv_{,x}=\bfu_{,t},\quad\bfR_{,tx}=\bfR_{,xt}.\label{eq2}
\eeq
By substituting (\ref{velo_v}) into the first, expanding the second,  and simplifying, we get
\beqs
&\bfW_2\bfv_0=(\bfW_1-\bfW)\bfu_0,&\label{eq1}\\
&\bfW_1\bfW_2-\bfW_2\bfW_1=\bfW_{1,x}-\bfW_{2,t}.&\label{w1w2}
\eeqs
Then both sides of \eqref{eq1} are aligned with the $\bfn_0$ direction or equal to zero. It follows that the axis of $\bfW_2$ and $\bfW_1-\bfW$ lie in the plane spanned by $\bfv_0$ and $\bfu_0$, and hence 
\beqs
\bfw_2&=&a(t,x)\bfv_0+b(t,x)\bfu_0, \label{axis1}\\
\bfw_1-\bfw&=&-b(t,x)\bfv_0+c(t,x)\bfu_0.\label{axis2}
\eeqs
for some $a,b,c\in\R$. Now we look at the initial conditions of $\bfW_1$ and $\bfW_2$. Since $\bfR(0,x)=\bfI$, then $\bfR_{,x}(0,x)=\bfR(0,x)\bfW_2(0,x)=\bfW_2(0,x)=\bf0$. 
This implies that $a(0,x)=b(0,x)=0$  and the left-hand side of \eqref{w1w2} is zero at $t=0$. From the right-hand side of \eqref{w1w2}, we have $c_{,x}(0,x)\bfu_0+c(0,x)\bfu_{0,x}=a_{,t}(0,x)\bfv_0+b_{,t}(0,x)\bfu_0$. Because $\bfu_{0,x}\cdot\bfn(0,x)=\bfy_{0,xx}\cdot\bfn(0,x)\neq0$, we get $c(0,x)=0$, and thus $\bfW_1(0,x)=\bfW$. To summarize, at $t=0$,
\beq
\bfW_1(0,x)=\bfW,\quad\bfW_2(0,x)=\bf0.\label{initial_ws}
\eeq

It is straightforward to verify that $(\bfW_1,\bfW_2)=(\bfW,\bf0)$ satisfies (\ref{eq1}-\ref{w1w2}). We now proceed to establish the uniqueness of this solution. Assume $(\bfW_1,\bfW_2)$ is another pair of solution distinct from $(\bfW,\bf0)$. From (\ref{initial_ws}), their initial conditions are $(\bfw_1(0,x),\bfw_2(0,x))=(\bfw,\bf0)$. Write \eqref{eq1} and \eqref{w1w2} in vector form and let $\Delta\bfw_1=\bfw_1-\bfw$ to get
\beqs
&\bfw_2\times\bfv_0=\Delta\bfw_1\times\bfu_0,&\label{50}\\
&(\Delta\bfw_1+\bfw)\times\bfw_2=\Delta\bfw_{1,x}-\bfw_{2,t}.&\label{51}
\eeqs
Consider the difference between the two solutions and derive an energy estimate,
\beq
E(t)=\int_{L_1}^{L_2} |\Delta\bfw_1|^2+|\bfw_2|^2\ dx,\quad {\rm with} \; E(0)=0.\label{energy}
\eeq
Then
\beqs
E_{,t}(t)&=&2\int_{L_1}^{L_2}(\Delta\bfw_1\cdot\Delta\bfw_{1,t}+\bfw_2\cdot\bfw_{2,t})\ dx.\label{energy_derivative}
\eeqs
Since the surface is assumed to be $C^3$-smooth on the compact domain $[0,T]\times[L_1,L_2]$, $\bfw_1$, $\bfw_2$, and their derivatives are bounded. Taking derivative of (\ref{50}) with respect to $t$, we have
\beq
\bfw_{2,t}\times\bfv_0=\Delta\bfw_{1,t}\times\bfu_0.
\eeq
We first assume $\bfw_{2,t}\neq0$ at time $t\in(0,t_0)$ for some $t_0<T$. There exists a bounded symmetric tensor $\bfM_1$, e.g., \beq
\bfM_1=\frac{\Delta\bfw_{1,t}\cdot\bfw_{2,t}}{\|\bfw_{2,t}\|^2}\bfI+\frac{\bfp\otimes\bfw_{2,t}+\bfw_{2,t}\otimes\bfp}{\|\bfw_{2,t}\|^2}
\eeq
with $\bfp=\Delta\bfw_{1,t}-\frac{\Delta\bfw_{1,t}\cdot\bfw_{2,t}}{\|\bfw_{2,t}\|^2}\bfw_{2,t}$, such that 
\beq
\Delta\bfw_{1,t}=\bfM_1\bfw_{2,t}=\bfM_1(\Delta\bfw_{1,x}-(\Delta\bfw_1+\bfw)\times\bfw_2).
\eeq
Then the first term of (\ref{energy_derivative}) can be reformulated to
\beqs
2\int_{L_1}^{L_2}\Delta\bfw_1\cdot\Delta\bfw_{1,t}dx&=&2\int_{L_1}^{L_2}\Delta\bfw_1\cdot\bfM_1(\Delta\bfw_{1,x}-(\Delta\bfw_1+\bfw)\times\bfw_2)dx\nonumber\\
&=&2\int_{L_1}^{L_2}\Delta\bfw_1\cdot\bfM_1\Delta\bfw_{1,x}dx-2\int_{L_1}^{L_2} \Delta\bfw_1\cdot\bfM_1((\Delta\bfw_1+\bfw)\times\bfw_2)dx\nonumber\\
&=&\int_{L_1}^{L_2}(\Delta\bfw_1\cdot\bfM_1\Delta\bfw_{1})_{,x}dx-\int_{L_1}^{L_2}\Delta\bfw_1\cdot\bfM_{1,x}\Delta\bfw_{1}dx-\nonumber\\
&&2\int_{L_1}^{L_2} \Delta\bfw_1\cdot\bfM_1(\Delta\bfw_1\times\bfw_2)dx-2\int_{L_1}^{L_2} \Delta\bfw_1\cdot\bfM_1(\bfw\times\bfw_2)dx\nonumber\\
&\leq&\Delta\bfw_1\cdot\bfM_1\Delta\bfw_{1}\big|_{x=L_1}^{x=L_2}+\int_{L_1}^{L_2} \|\bfM_{1,x}\|\ |\Delta\bfw_1|^2dx\nonumber\\
&&+\int_{L_1}^{L_2}C_1\|\bfM_1\|( |\Delta\bfw_1|^2+|\bfw_2|^2)dx+\int_{L_1}^{L_2}\|\bfM_1\| |\bfw|( |\Delta\bfw_1|^2+|\bfw_2|^2)dx\nonumber\\
&=&\int_{L_1}^{L_2} \|\bfM_{1,x}\|\ |\Delta\bfw_1|^2dx+\int_{L_1}^{L_2}C_1\|\bfM_1\|( |\Delta\bfw_1|^2+|\bfw_2|^2)dx\nonumber\\
&&+\int_{L_1}^{L_2}\|\bfM_1\| |\bfw|( |\Delta\bfw_1|^2+|\bfw_2|^2)dx.\label{first_term}
\eeqs
Here $C_1$ denotes the upper bound of $|\Delta \bfw_1|$. In deriving this estimate, we have applied the Cauchy-Schwarz inequality ($\bfa\cdot \bfb \le |\bfa|\ |\bfb|$), Young’s inequality ($\bfa\cdot \bfb \le \tfrac{1}{2}(|\bfa|^2 + |\bfb|^2)$), and the operator norm bound ($|\bfM\bfv| \le \|\bfM\|_2\ |\bfv|$). The boundary term vanishes under the assumption that the curve is closed. For the second term of (\ref{energy_derivative}), we have
\beqs
2\int_{L_1}^{L_2}\bfw_2\cdot\bfw_{2,t}dx&=&2\int_{L_1}^{L_2}\bfw_2\cdot(\Delta\bfw_{1,x}-(\Delta\bfw_1+\bfw)\times\bfw_2)dx=2\int_{L_1}^{L_2}\bfw_2\cdot\Delta\bfw_{1,x}dx.
\eeqs
Taking derivative of (\ref{50}) with respect to $x$, we have 
\beq
\Delta\bfw_{1,x}\times\bfu_0=\bfw_{2,x}\times\bfv_0+\bfw_{2}\times\bfv_{0,x}-\Delta\bfw_{1}\times\bfu_{0,x}.
\eeq
Since $\bfw_{2,t} \neq 0$, we have $\bfw_2 \neq 0$ for $t \in (0,t_0)$. Therefore, at least one of $\bfw_{2,x}, \bfw_2$, or $\Delta \bfw_1$ must be nonzero. By a similar argument, there exist bounded symmetric tensors $\bfM_2$, $\bfM_3$, and $\bfM_4$ such that
\beq
\Delta\bfw_{1,x}=\bfM_2\bfw_{2,x}+\bfM_3\bfw_2+\bfM_4\Delta\bfw_1.
\eeq
Then the second term gives
\beqs
2\int_{L_1}^{L_2}\bfw_2\cdot\Delta\bfw_{1,x}dx&=&2\int_{L_1}^{L_2}\bfw_2\cdot(\bfM_2\bfw_{2,x}+\bfM_3\bfw_2+\bfM_4\Delta\bfw_1)dx\nonumber\\
&=&\bfw_2\cdot\bfM_2\bfw_{2}\big|_{x=L_1}^{x=L_2}-\int_{L_1}^{L_2}\bfw_2\cdot\bfM_{2,x}\bfw_{2}dx+2\int_{L_1}^{L_2}(\bfw_2\cdot\bfM_3\bfw_2+\bfw_2\cdot\bfM_4\Delta\bfw_1)dx\nonumber\\
&\leq&\int_{L_1}^{L_2}\|\bfM_{2,x}\|\ |\bfw_2|^2dx+2\|\bfM_3\|\ |\bfw_2|^2+\|\bfM_4\|(|\Delta\bfw_1|^2+|\bfw_2|^2)dx.\label{second_term}
\eeqs
Combining \eqref{first_term} and \eqref{second_term} and collecting terms, we obtain an inequality of Gronwall type with some constant $C$,
\beq
E_{,t}(t)\leq C\int_{L_1}^{L_2} |\Delta\bfw_1|^2+|\bfw_2|^2\ dx= CE(t),\quad {\rm with} \; E(0)=0.
\eeq
Solve the inequality to get
\beq
E(t)\leq E(0)\exp{(tC)}=0.
\eeq
Since the integrand in \eqref{energy} is nonnegative, we get
\beq
\Delta\bfw_1=\bfw_2=\bf0.
\eeq
If $\bfw_{2,t}=\bf0$ at $t\in (0,t_0)$, then the initial condition implies $\bfw_2=\bf0$. Moreover, from (\ref{axis1}-\ref{axis2}) and (\ref{50}-\ref{51}), it follows that $\Delta\bfw_1=\bf0$. In this case, the solution at $(0,t_0)$ coincides with the known solution. Applying the same argument on the subsequent interval $t>t_0$,  we conclude that $(\bfW,\bf0)$ is the unique solution. Hence, $\bfR$ is independent of $x$ and is determined by
\beq
\bfR_{,t}(t)=\bfR(t)\bfW,\quad \bfR(0)=\bfI.
\eeq
This corresponds to a rotation about the fixed axis $\bfw$ with angle $t$. Together with \eqref{velo_v}, we have 
\beq
\bfv(t,x)=(\log\lambda\bfI+\bfW)\lambda^t\bfR(t)\bfy_0(x),
\eeq
and $\bfy(t,x)$ is found by
\beq
\bfy(t,x)=\bfy_0(x)+\int_0^t\bfv(r,x) dr
=\lambda^t\bfR(t)\bfy_0(x),\quad x\in[L_1,L_2].
\eeq
\end{proof}

\begin{remark} \label{remark1}
The theorem above concerns shells whose growing edge is a closed curve.  If the growing edge is not closed, (\ref{growth_laws}) remains necessary and sufficient for generating a shell of the form (\ref{manifold}).  The necessity follows by the same argument, and sufficiency is established by verifying that $(\bfW,\bf0)$ is a solution of (\ref{eq1}, \ref{w1w2}).  However, we are unable to prove that (\ref{manifold}) is the unique solution.
\end{remark}


\section{Ornamentation patterns}
{ The Lie group $\calG$ used to generate the main shell surfaces shown in Figure \ref{spine}A is given by (\ref{group}), with $\lambda=1.10$, $1.08$, and $1.24$ for the three examples, respectively. The positions of the ornamentation patterns are determined by discrete subgroups of the corresponding $\calG$, namely
\beqs
\calG_1&=&\{\lambda^t\bfR(t), ; t\in0.75\mathbb{Z}\},\label{g1}\\
\calG_2&=&\{\lambda^t\bfR(t), ; t\in2.04\mathbb{Z}\},\\
\calG_3&=&\{\lambda^t\bfR(t), ; t\in0.34\mathbb{Z}\}
\eeqs
where $\mathbb Z$ denotes the integers.  To illustrate the construction, we consider the axial lamellae in Epitoniidae in Figure \ref{spine}A as an example. The discrete subgroup basically determines the spatial arrangement of the objects to which it is applied. For instance, applying the discrete group (\ref{g1}) to a point $\bfy_0(x_0)$ on the initial seed curve $\bfy$ will generate a collection of points in purple shown in Figure \ref{frond1}(A). To construct the axial lamella, however, we apply $\calG_1$ not to a single point but to a surface generated by $\calG_f(\bfy_0)$, as illustrated in the left panel of Figure \ref{frond1}(B). In this example, $\calG_f$ is given by
\beq
\calG_f=\{\lambda_f^t\bfR(t),;\lambda_f=2.69,;t\in[0,0.15]\}.
\eeq
The resulting ornamentation pattern is then obtained as
\beq
\calS_f=\calG_1(\calG_f(\bfy_0(x))),\quad x\in[L_1,L_2],
\eeq
which is shown in the right panel of Fig.~\ref{frond1}(B).
\begin{figure}[h!]
\centering
\includegraphics[width=0.8\textwidth]{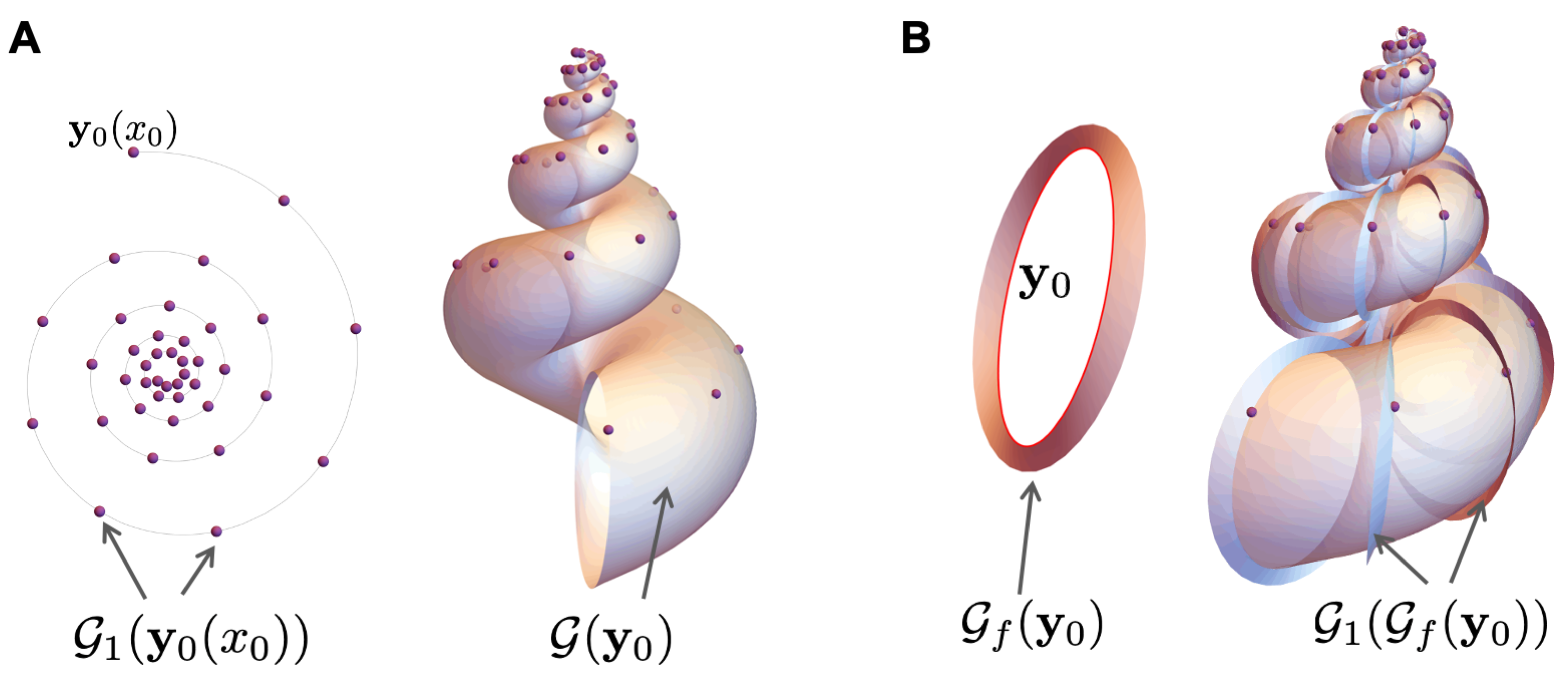}
\caption{An illustration of how a discrete subgroup generates ornamentation patterns. (A) Top view of the orbit generated by applying the subgroup $\calG_1$ to a point $\bfy_0(x_0)$, together with the corresponding discrete locations on the shell surface $\calG(\bfy_0)$. (B) A single axial lamella in Epitoniidae generated by a distinct Lie group $\calG_f$ acting on $\bfy_0$, and the complete ornamentation pattern obtained by applying the subgroup $\calG_1$ to $\calG_f(\bfy_0)$.}
\label{frond1}
\end{figure} 
}

\end{document}